\documentclass[11pt]{article}
\usepackage[T1]{fontenc}
\usepackage[utf8]{inputenc}
\usepackage[margin=1in]{geometry}

\usepackage{amsfonts,amsthm,amsmath,amssymb}
\usepackage{cite,microtype,graphicx,hyperref}
\usepackage[capitalize,nameinlink]{cleveref}

\newtheorem{theorem}{Theorem}
\newtheorem{algorithm}{Algorithm}
\newtheorem{lemma}{Lemma}

\title{Finding Relevant Points for Nearest-Neighbor Classification}
\author{David Eppstein\thanks{Department of Computer Science, University of California, Irvine. This research was performed in part through the 1st and 2nd Virtual Workshops on Computational Geometry, organized in 2020 and 2021 by Erik Demaine. We thank Erik and the other workshop participants for hosting these events and for bringing our attention to these problems.}}

\date{ }

\begin{document}
\maketitle  

\begin{abstract}
In nearest-neighbor classification problems, a set of $d$-dimensional training points are given, each with a known classification, and are used to infer unknown classifications of other points by using the same classification as the nearest training point. A training point is \emph{relevant} if its omission from the training set would change the outcome of some of these inferences. We provide a simple algorithm for thinning a training set down to its subset of relevant points, using as subroutines algorithms for finding the minimum spanning tree of a set of points and for finding the extreme points (convex hull vertices) of a set of points. The time bounds for our algorithm, in any constant dimension $d\ge 3$, improve on a previous algorithm for the same problem by Clarkson (FOCS 1994).
\end{abstract}

\section{Introduction}
Nearest-neighbor classification is a widely-used supervised learning technique in which, from training data consisting of geometric points with discrete classifications, one infers the classification of new points as equal to the classification of the nearest training point~\cite{CovHar-TIT-67}. This technique has motivated many important developments in exact and approximate nearest-neighbor searching, including the construction of Voronoi diagrams~\cite{AurKleLee-13,Cha-DCG-93,Dwy-DCG-91,GuiKnuSha-Algo-92,Wat-CJ-81}, the development of point-location data structures for performing queries in these diagrams~\cite{GuiKnuSha-Algo-92,PreTam-SICOMP-92}, quadtree-based data structures for approximate nearest neighbors in spaces of moderate dimension~\cite{AryMouNet-JACM-98,AryMalMou-JACM-10,Cha-DCG-98,EppGooSun-IJCGA-08}, and locality-sensitive hashing for higher-dimensional data~\cite{IndMot-STOC-98,GioIndMot-VLDB-99,DatImmInd-SoCG-04,AndInd-CACM-08}.

Voronoi diagrams have high worst-case complexity even for spaces of moderate dimensions: for $n$ points in $d$ dimensions, the Voronoi diagram can have complexity $\Theta(n^{\lceil d/2\rceil})$~\cite{DewVra-UM-77,Kle-AM-80,Sei-AGDM-91}. Approximate nearest-neighbor searching data structures are better-behaved, but high-dimensional classification problems using them still have a complexity only a small factor faster than a naive scan of all training data for each new input point. To speed up these methods, it is natural to consider a preprocessing stage that reduces the training set to a smaller subset, its \emph{relevant points}, before building these data structures~\cite{Cla-FOCS-94,BreDemEri-DCG-05}. Here, we define a point to be relevant if it is needed for correct nearest-neighbor classification: removing it would change the nearest-neighbor classification of some points of $\mathbb{R}^d$. Removing all points that are not relevant leaves all nearest-neighbor classifications unchanged (see \cref{thm:correctness}). Although this has no benefit in the worst case, it is reasonable to expect that on average for smooth-enough input distributions and smooth-enough decision boundaries (we leave those terms deliberately vague and non-rigorous) a training set of $n$ points in $d$ dimensions may be reduced to a smaller subset of $O(n^{(d-1)/d})$ relevant points, which could in some cases be a significant savings.

In this work we investigate a simple algorithm for quickly identifying which points of a training set are relevant and which are not. Our algorithm is \emph{output-sensitive}: it is faster when there are few relevant points, and slower when there are many. It is based on the solutions to two previously-studied geometric problems: the construction of Euclidean minimum spanning trees, spanning trees for the complete graph on a given set of points, weighted by the Euclidean distances between pairs of points, and the identification of extreme points, the vertices of the convex hull of a given set of points.

\subsection{New results}
Given a set of training points with discrete classifications (not assumed to be binary nor in general position),
our algorithm performs the following steps, using a single Euclidean minimum spanning tree construction followed by an extreme-point computation per relevant point.  As we prove, it finds exactly the set of relevant points.

\begin{algorithm}[relevant points of a training set]
\label{alg:relevant}
~\\[-3ex]
\begin{enumerate}
\item Find a Euclidean minimum spanning tree $T$ of the training set.
\item Find the edges in $T$ whose two endpoints have different classifications, and initialize the set $R$ of relevant points to consist of the endpoints of these edges.
\item For each relevant point $r$ added to $R$ (either initially or within this loop), perform the following steps:
\begin{enumerate}
\item Invert through a unit sphere centered at $r$ all of the training points whose classification differs from $r$, producing a point set $S_r$, including also $r$ itself in $S_r$.
\item Identify the extreme points of $S_r$.
\item Add to $R$ the training points corresponding to extreme points of $S_r$.
\end{enumerate}
\end{enumerate}
\end{algorithm}

\begin{figure}[t]
\includegraphics[width=\textwidth]{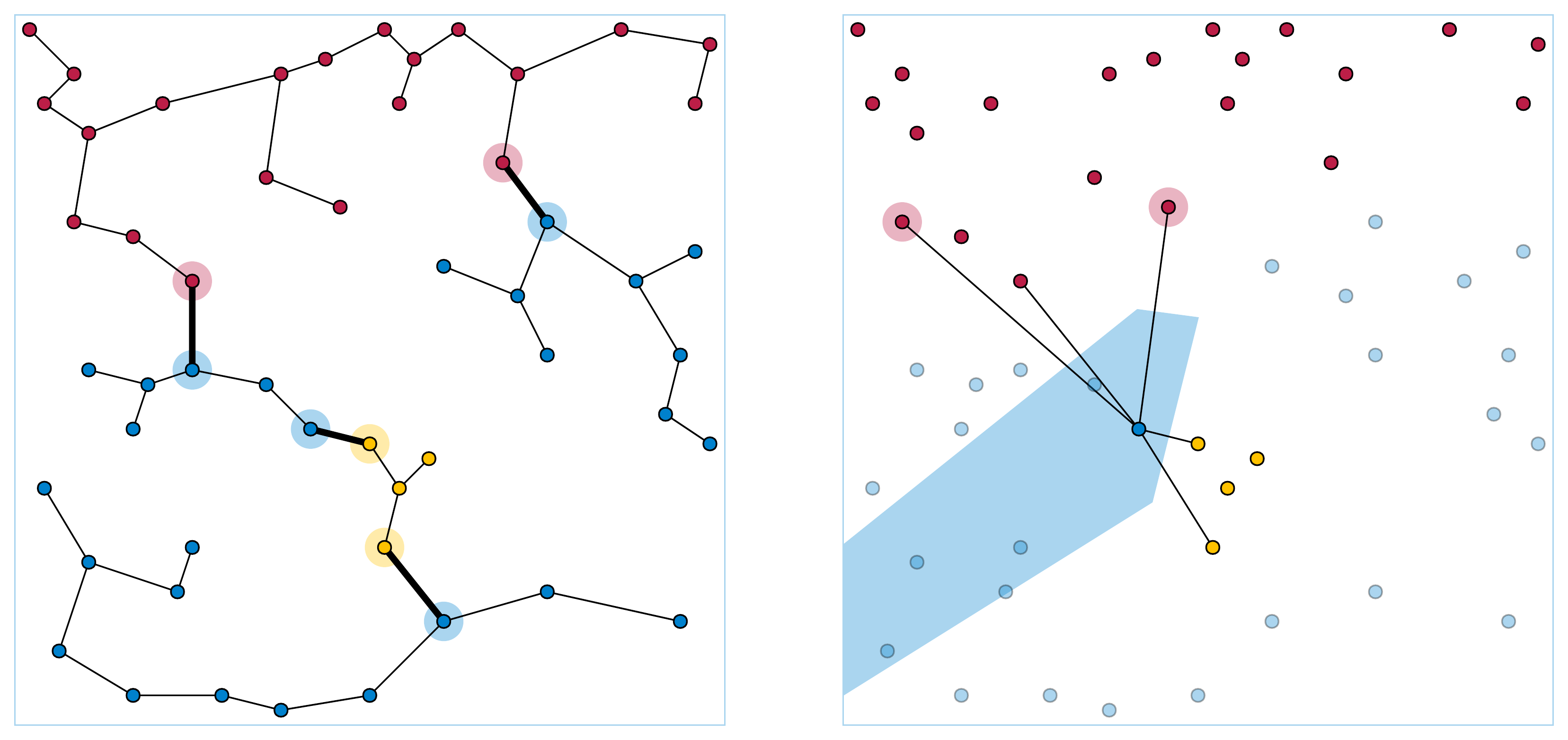}
\caption{The steps of \cref{alg:relevant}. Left: A training set in $\mathbb{R}^2$ with points classified as red, blue, and yellow, and its minimum spanning tree, with the eight highlighted points identified as relevant because they are endpoints of bichromatic edges. Right: The same training set, showing one of the previously-identified relevant points $r$ (the single dark blue point), the set $S_r$ (all dark points), the cell for $r$ in the Voronoi diagram of $S_r$ (light blue shaded region), and the five Voronoi neighbors of $r$. Three of those neighbors were already identified as relevant from the minimum spanning tree; the other two neighbors, identified as relevant in the second part of the algorithm, are highlighted.}
\label{fig:steps}
\end{figure}

The steps of the algorithm are illustrated in \cref{fig:steps}. The intuition behind \cref{alg:relevant} is that the minimum spanning tree phase of the algorithm finds a piece of each component of the decision boundary, a wall between the Voronoi cells of two relevant points. The extreme point phase finds the neighbors of each relevant point in the Voronoi diagram of $S_r$, including the relevant points that define neighboring walls of the decision boundary. In this way, it expands each piece of boundary to the full component of the boundary, from each wall to its adjacent walls in the component, without finding any false positives, and without needing to know anything about the topology of the boundary. The algorithm's efficiency comes from output-sensitivity both in the number of calls to the extreme-point subproblem and within the algorithms for this subproblem. Both the Euclidean minimum spanning tree and the extreme-point subproblem admit either simple and dimension-independent algorithms or asymptotically faster but more dimension-specific and more complex algorithms, but unfortunately we do not know of algorithms that are both simple and optimally efficient. We analyze the same overall algorithm both ways, using both kinds of  algorithm for the subproblems. This analysis gives us the following results:

\begin{itemize}
\item Using simple algorithms, for an input of bounded dimension with $k$ relevant points, we can identify the relevant points in time $O(n^2+k^2n)$. For an input of unbounded dimension, we can identify the relevant points using $O(kn)$ linear programs.
\item Using more complex subroutines for the Euclidean minimum spanning tree and extreme points, we can identify the relevant points in randomized expected time $O\bigl((n\log n)^{4/3}+kn\log k\bigr)$ for three-dimensional points, and in time
\[
O\left(n^{2-\frac{2}{\lceil d/2\rceil+1}+\varepsilon}
+k(nk)^{1-\frac{1}{\lfloor d/2\rfloor+1}}(\log n)^{O(1)}\right).
\]
for $d$-dimensional points, for any $\varepsilon>0$. For instance, for $d=4$ this bound is $O(n^{4/3+\varepsilon}+k^{5/3}n^{2/3}\log^{O(1)}n)$, and for $d=5$ it is $O(n^{3/2+\varepsilon}+k^{5/3}n^{2/3}\log^{O(1)}n)$.
\end{itemize}

The nearest-neighbor decision boundary, and not just the set of relevant points, can be constructed as a sub-complex of the Voronoi diagram of the relevant points, consisting of the $(d-1)$-dimensional faces of the diagram that separate cells of opposite classifications. This construction can be performed using standard Voronoi diagram construction algorithms~\cite{Cha-DCG-93} in an additional time of $O(k\log k + k^{\lceil d/2\rceil})$, independently of~$n$.
For all dimensions greater than two these time bounds are significantly faster than the $O(n^{\lceil d/2\rceil})$ time that could be obtained by constructing the Voronoi diagram directly and then using its faces to identify the classification decision boundary.

We note that the set of relevant points found by \cref{alg:relevant} is not necessarily the smallest set of points that would have the same decision boundary, or the smallest subset of the given points that would have the same decision boundary. It is, rather, the set of all training points whose omission from the whole data set (or from the resulting subset of relevant points) would change the decision boundary. Finding the smallest set of points with the same decision boundary, in high dimensions, seems likely to be a much more difficult task.

\subsection{Related work}

The problem of constructing nearest-neighbor decision boundaries, in an output-sensitive way, was considered by Bremner et al.~\cite{BreDemEri-DCG-05}, for the special case of two-dimensional training data. They showed that, for training sets of this dimension, the relevant points and their resulting nearest-neighbor decision boundary can be found in time $O(n\log k)$.

In higher dimensions, the only work we are aware of for this problem is that of Clarkson~\cite{Cla-FOCS-94}, who (in our terms) gave a simple algorithm for finding the relevant points whose running time is $O\bigl(\min(n^3,kn^2\log n)\bigr)$ whenever the dimension is bounded. Our time bounds are significantly faster than Clarkson's in the cases for which this sort of thinning is particularly useful, when $k$ is much smaller than~$n$.

We will survey algorithms for the two subroutines we use, for Euclidean minimum spanning trees and extreme points, in our discussion of these problems in the next section.

\section{Preliminaries}

\subsection{Voronoi diagrams and Delaunay graphs}

The Voronoi diagram of a finite set $S$ of points in $\mathbb{R}^d$ (called \emph{sites} in this context) is a collection of convex polyhedra (possibly unbounded), one for each site site $s$, consisting of the points in $\mathbb{R}^d$ for which $s$ is a nearest site, one with minimum Euclidean distance to the point~\cite{AurKleLee-13}. We call such a polyhedron the \emph{cell} of $s$. It is an intersection of a finite system of halfspaces, the halfspaces that contain $s$ and have as their boundaries the hyperplanes halfway between $s$ and each other site. These hyperplanes are the perpendicular bisectors of line segments connecting $s$ to each other site. For our purposes it is convenient to think of the cells as closed sets, containing their boundaries, and intersecting each other at shared boundary points. Their interiors, however, are disjoint. The union of all the cells equals  $\mathbb{R}^d$.

The Voronoi diagram has finitely many \emph{faces}, the intersections of finite sets of cells. The dimension of a face is the dimension of its affine hull. We define a \emph{wall} of the Voronoi diagram to be a face of dimension $d-1$, and a \emph{junction} of the diagram to be a face of dimension $d-2$.
The affine hull of a junction, a $(d-2)$-dimensional subspace of $\mathbb{R}^d$, is perpendicular to a family of two-dimensional planes. If we choose a plane in this family that intersects the junction in its relative interior, a point that is not part of any lower-dimensional face, then the junction will appear in this intersection as a point, and in a neighborhood of this point the walls that include the junction will appear as rays and the cells that include the junction will appear as convex wedges between these rays. The geometry of this structure of rays and wedges does not depend on the choice of intersecting plane. It is called the \emph{link} of the junction. If the sites are in general position (no $d+2$ of them belonging to a common sphere) then the link of a junction will only have three rays and three wedges, but we do not wish to assume general position.

We define the \emph{Delaunay graph} of a set of sites to be a graph having the sites as vertices, with edges connecting pairs of sites when their cells intersect in a wall. We do not include edges for pairs of sites whose cells have a lower-dimensional or empty intersection. For sites in general position, this graph forms the set of edges of a simplicial complex, the \emph{Delaunay triangulation}, but again we do not wish to make this general position assumption.

\subsection{Euclidean minimum spanning trees}

A Euclidean minimum spanning tree of a set $S$ of point sites is just a minimum spanning tree of a complete graph, having $S$ as its vertices, with edges weighted by Euclidean distance. We allow equal distances, in which case there may be more than one possible minimum spanning tree. If there are $n$ sites, a minimum spanning tree can be constructed easily by naive algorithms in time $O(n^2)$~\cite{Epp-HCG-00}: simply construct the weighted complete graph, and apply either Borůvka's algorithm or Jarník's algorithm (with an unordered list as priority queue), both of which take time $O(n^2)$ for dense graphs.

For points in $\mathbb{R}^2$, a standard and more efficient method of constructing a Euclidean minimum spanning tree is to construct the Delaunay triangulation of the points (perturbed if necessary to be in general position), which is guaranteed to contain a minimum spanning tree as a subgraph~\cite{ShaHoe-FOCS-75}, and then apply a planar graph minimum spanning tree algorithm to the resulting graph. This method provides no advantage for worst-case analysis in higher dimensions, as the Delaunay graph can be complete~\cite{Wat-CJ-81}, but we need the following related lemma for the correctness of \cref{alg:relevant}:

\begin{lemma}
\label{lem:mst-is-delaunay}
For an arbitrary finite set $S$ of sites in $\mathbb{R}^d$, any Euclidean minimum spanning tree of the sites is a subgraph of the Delaunay graph of the sites.
\end{lemma}

\begin{proof}
Consider any minimum spanning tree edge $pq$, and let $r$ be the midpoint of the edge between sites $p$ and $q$. Then $r$ is equidistant from $p$ and $q$. If any other site $s$ were not strictly farther than $r$, then (by the triangle inequality) $ps$ and $qs$ would both be shorter than $pr$, allowing the tree to be made shorter by replacing edge $pq$ by $ps$ or $qs$ (whichever is not already in the tree). Because we are assuming the tree to be minimum, shortening it is not possible, so all other sites must be strictly farther from $r$.

Let $r'$ be obtained by perturbing $r$ by a small amount in any direction perpendicular to $pq$; this perturbation preserves the property that $r'$ is equidistant from $p$ and $q$ and farther from all other sites. Therefore, $r'$ lies on a wall between the cells for $p$ and $q$, and edge $pq$ is also an edge of the Delaunay graph.
\end{proof}

Instead, specialized higher-dimensional Euclidean minimum spanning tree algorithms proceed by reducing the problem to a collection of bichromatic closest pair problems, in which one must find the closest red--blue pair among a collection of red and blue points, combining the resulting pairs into a graph, and applying a graph minimum spanning tree algorithm to this graph. Several reductions from Euclidean minimum spanning trees to bichromatic closest pairs have been given but for the known time bounds for bichromatic closest pairs these reductions all have the same efficiency~\cite{AgaEdeSch-DCG-91,CalKos-SODA-93,KrzLevNil-NJC-99}. Based on this approach, the following results are known~\cite{AgaEdeSch-DCG-91}:\footnote{Agarwal et al.~\cite{AgaEdeSch-DCG-91} state the time bound for high-dimensional Euclidean minimum spanning trees as a randomized expected time bound, but in later related work such as \cite{AgaMat-Algo-95} they observe that the need for randomness can be eliminated using techniques from \cite{Mat-JCSS-95}.}

\begin{lemma}
A 3-dimensional Euclidean minimum spanning tree of $n$ points can be computed in randomized expected time $O\bigl((n\log n)^{4/3}\bigr)$. A $d$-dimensional Euclidean minimum spanning tree can be computed by a deterministic algorithm in time
\[
O\left(n^{2-\frac{2}{\lceil d/2\rceil+1}+\varepsilon}\right)
\]
for any $\varepsilon>0$.
\end{lemma}

When the dimension is not constant, these methods fail to improve on the quadratic time of the naive algorithms. More strongly, for any $\varepsilon>0$, the strong exponential time hypothesis implies that closest pairs and therefore also Euclidean minimum spanning trees of dimension polylogarithmic in~$n$ (with the polylogarithm depending on~$\varepsilon$) cannot be found in time $O(n^{2-\varepsilon})$~\cite{KarMan-Comb-20}.

\subsection{Extreme points}

The \emph{extreme points} of a finite set $S$ of points are the vertices of its convex hull, or equivalently the points that are on the boundary of a halfspace in which all other points of $S$ are interior.\footnote{See, for instance, \cite{Gru-CP-03}, p.~35, where this equivalence is stated in the form that the extreme points and exposed points of a convex polytope coincide.} Testing whether a given point $p$ is extreme can be formulated as a linear programming feasibility problem in which we seek a vector $v$ for which $v\cdot p>v\cdot q$ for all other points $q$ in $S$. The point $p$ is extreme if and only if such a vector~$v$ exists.

The dimension $d$ equals the number of variables in this linear program. Therefore, when $d$ is bounded, we can apply algorithms for low-dimensional linear programming, which are strongly polynomial and take linear time when the dimension is bounded, with exponential or subexponential dependence on the dimension~\cite{Meg-JACM-84,DyeFri-MP-89,Sei-DCG-91,Cla-JACM-95,MatShaWel-Algo-96,Cha-TALG-18}. One particularly simple choice here is Seidel's algorithm, which considers the constraints of the given program in a random order, maintaining the optimal solution for the constraints seen so far, and when finding a violated constraint recurses within the subspace of one lower dimension in which that constraint is tight~\cite{Sei-DCG-91}.

Thus, one simple way of finding all of the extreme points is to apply this linear programming approach to each point individually. However, this can be improved by the following method:

\begin{algorithm}[simple algorithm for extreme points~\cite{Cla-FOCS-94,Cha-DCG-96a,OttSchSou-NJC-01}]
\label{alg:simple-extreme}
~\\[-3ex]
\begin{enumerate}
\item Initialize the set $E$ of extreme points to the empty set.
\item For each input point $p$, in an arbitrary order, perform the following steps:
\begin{enumerate}
\item Use the linear program outlined above to test to test whether $p$  is extreme for $E\cup\{p\}$ and find a vector $v$ whose dot product with $v$ exceeds the dot product with any point in $E$.
\item If $p$ is not extreme for $E\cup\{p\}$ ($v$ does not exist) go on to the next point in the outer loop.
\item Otherwise, find the training point $p'$ that maximizes $v\cdot p'$, breaking ties lexicographically by the coordinates of the training points, and add $p'$ to $E$.
\end{enumerate}
\end{enumerate}
\end{algorithm}
Each iteration of the outer loop that reaches step (c) identifies a new extreme point.
Thus, all extreme points are identified using $n$ linear programs of size $O(k)$, plus $k$ extreme-point searches of size $n$, in total time $O(kn)$ in bounded dimensions. We omit some details and any proof of correctness, for which see the references for this algorithm~\cite{Cla-FOCS-94,Cha-DCG-96a,OttSchSou-NJC-01}. This $O(kn)$ time bound can be further improved at the cost of greater algorithmic complexity. When the dimension is two or three, and $k$ of the $n$ given points are extreme, one can find all the extreme points in time $O(n\log k)$ using an output-sensitive algorithm for the convex hull, as (by Euler's polyhedral formula) the number of extreme points and convex hull complexity are always within a constant factor of each other~\cite{KirSei-SICOMP-86,ClaSho-DCG-89,ChaMat-CG-95,Cha-DCG-96b}. In higher dimensions, Chan~\cite{Cha-DCG-96b} gives a time bound for this problem of
\[
O\left(n(\log k)^{O(1)}+(nk)^{1-\frac{1}{\lfloor d/2\rfloor+1}}(\log n)^{O(1)}\right).
\]

\subsection{Inversion and polar duality}

An \emph{inversion} through a sphere $S$ in Euclidean space $\mathbb{R}^d$ is a transformation that maps points (other than the center of the circle) to other points. Each point and its transformed image lie on a common ray from the center of the circle, and the product of their distances from the center equals the squared radius of the circle. These transformations preserve cocircularity of the transformed points, but not other properties such as distances. \emph{Polarity} is a different kind of transformation, again defined by a sphere in $\mathbb{R}^d$, that associates points (other than the center of the sphere) with hyperplanes (not through the center of the sphere) and vice versa. The point associated with a hyperplane $H$ can be found by finding the point $p$ that belongs to $H$ and is nearest  to the center of the sphere, and then inverting $p$ through the sphere. Reversing these steps, the hyperplane associated with a point $p$ is the hyperplane through the inverted image of $p$, perpendicular to the line through $p$ and the center of the sphere. Both inversion and polarity are illustrated in \cref{fig:inverse-polar}.

\begin{figure}[t]
\includegraphics[width=\textwidth]{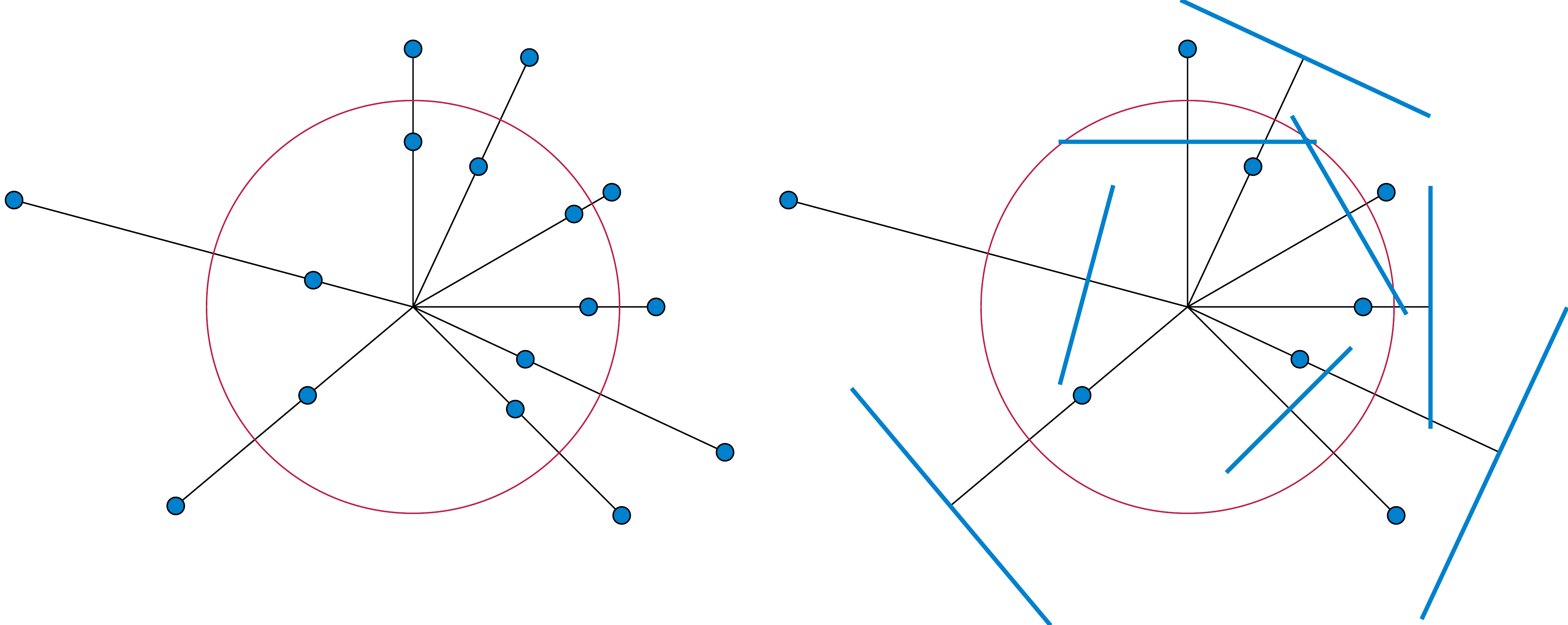}
\caption{Inverse pairs of points with respect to a circle (left), and polar pairs of points and lines with respect to a circle (right).}
\label{fig:inverse-polar}
\end{figure}

We can choose a Cartesian coordinate system for which the given sphere is the unit sphere centered at the origin. Under these coordinates, if $p$ is a point and $q$ is the polar image of a hyperplane $H$, then $p$ lies on $H$ if and only if $p\cdot q=1$. If the dot product is less than one, $p$ lies on the same side of $H$ as the origin, and if the dot product is greater than one, $p$ lies on the far side of $H$ from the origin. Because the dot product is a commutative operation, the operation of polarity (taking $p$ to a hyperplane and $H$ to the point $q$) preserves incidence and sidedness.

The inversions performed in our algorithm for finding relevant points can alternatively be thought of as polarities, with the inverted image of each site representing the polar image of the hyperplane equidistant between it and the chosen relevant point~$r$. We can formalize this intuition in the following lemma:

\begin{lemma}
\label{lem:extreme-wall}
Let $P=\{p_1,p_2,\dots\}$ be a set of points, let $r$ be a point not belonging to $P$,
and let $Q=\{q_1,q_2,\dots\}$ be the set of points obtained by inverting $P$ through a sphere centered at $r$, with corresponding indexes. Then $q_i$ is extreme in $Q\cup\{r\}$ if and only if the cells for $p_i$ and $r$ share a wall in the Voronoi diagram of $P\cup\{r\}$.
\end{lemma}

\begin{proof}
Consider any particular point $q_i$, for which we wish to prove the statement of the lemma. Scaling the radius of the sphere of inversion scales $Q\cup\{r\}$ but does not change the property of being extreme, so the choice of radius is irrelevant to the truth of the lemma. Therefore, without loss of generality, we may assume that the radius is such that the polar hyperplane of $q_i$ (for a different sphere of unit radius centered at $r$) is exactly the wall $W_i$ of the two-point Voronoi diagram of $\{r,p_i\}$.

If $q_i$ is extreme, there is a hyperplane $H$ passing through it, with all of the other points in $Q\cup\{r\}$ on a single side of $H$. Let $h$ be the polar of $H$; then $h$ lies on $W_i$, so it is equidistant from $r$ and $p_i$. By the preservation of sidedness of polarity, $h$ lies on the same side as $r$ of each wall $W_j$ for $j\ne i$, so it is farther from all the other points $p_j$ than it is from $r$ and from $p_i$. Because it is equidistant from $r$ and $p_i$ and farther from all the other sites of the Voronoi diagram, it witnesses the existence of a Voronoi wall between $r$ and $p_i$ in the full Voronoi diagram.

In the opposite direction, if the cells for $p_i$ and $r$ share a wall in the Voronoi diagram of $P\cup\{r\}$, let $h$ be a point in the relative interior of that wall, and let $H$ by the hyperplane polar to $h$. Because $h$ is equidistant from $r$ and $p_i$, and farther from all of the other points $p_j$, $h$ is on the same side as $r$ of each wall $W_j$. By the preservation of sidedness of polarity, point $q_i$ lies on $H$, with all other points $q_j$ on the same side as $r$ of $H$, so $H$ witnesses the fact that $q_i$ is extreme in $Q\cup\{r\}$.
\end{proof}

\subsection{Binary homology}

\begin{figure}[t]
\centering
\includegraphics[width=0.8\textwidth]{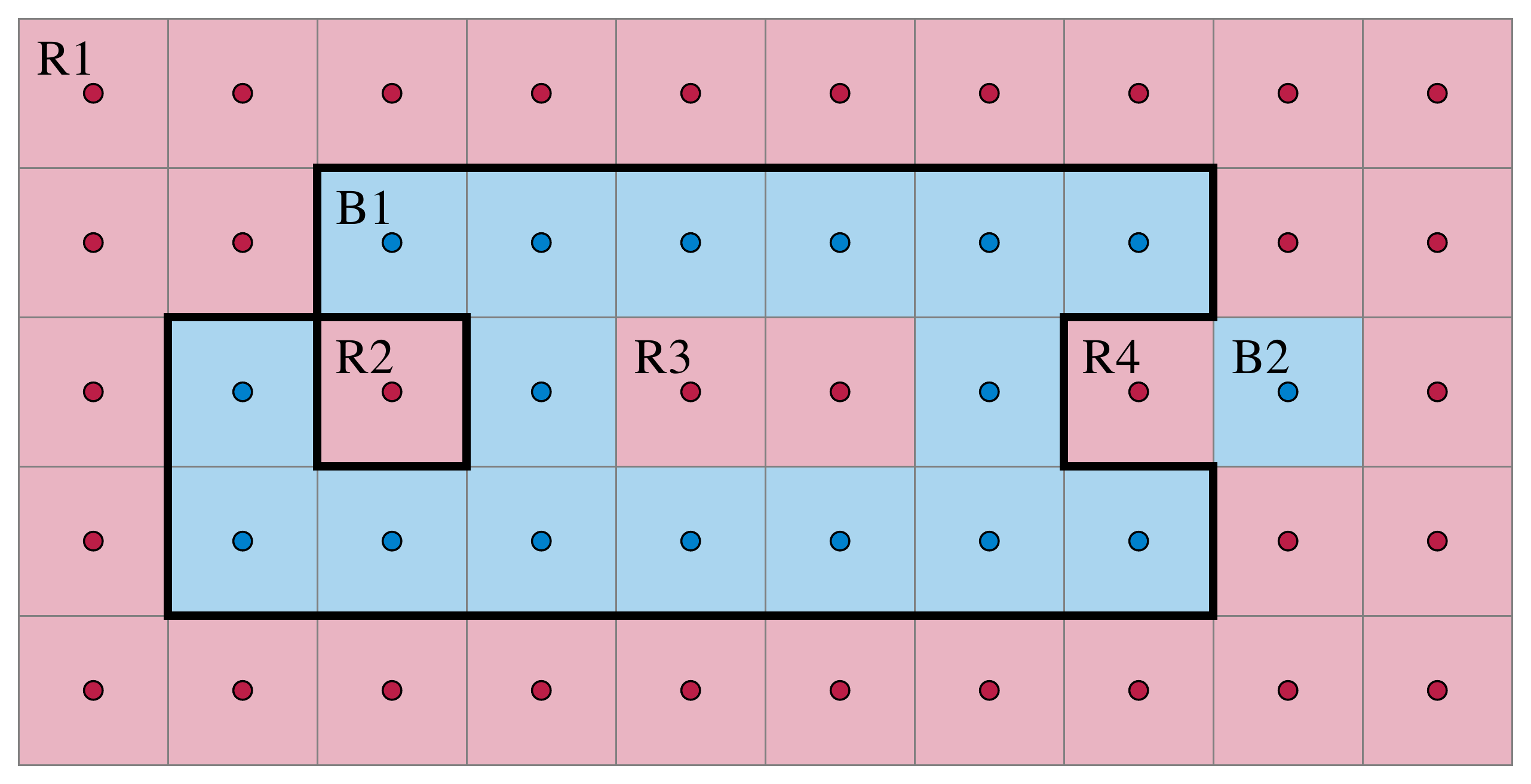}
\caption{The Voronoi diagram of a training set in the form of a $5\times 10$ grid, with its six decision components (R1, R2, R3, R4, B1, and B2) labeled in their upper left Voronoi cells. Decision component B1 has two decision boundary components: the one highlighted, between B1, R1, R2, and R4, and a second boundary component between B1 and R3. The highlighted decision boundary component is itself a boundary, of B1${}\cup{}$R3 or of its complement.}
\label{fig:decision-boundaries}
\end{figure}

Although \cref{alg:relevant} itself is ignorant of topology, we need some basics for its correctness proof. Specifically, we use mod-2 homology, as described e.g. by \cite{Hen-94}. This theory applies to a wide class of cell complexes, but we do not need to define this class carefully, because we will only apply this form of homology to the finite convex subdivisions of $\mathbb{R}^d$ obtained as Voronoi diagrams of finite point sets.

If $S$ is any set, the family of all subsets of $S$ forms a $\mathbb{Z}_2$-vector space with the symmetric difference of sets as its vector addition operation. Subsets of the $k$-dimensional faces of a polyhedral decomposition of space (such as a Voronoi diagram) are called $k$-chains. For each $k$, the $k$-chains can be mapped to $(k-1)$-chains by the \emph{boundary map} $\partial_k$, which takes a single $k$-face to the $(k-1)$-faces on its boundary, and acts linearly on $k$-chains: if $C$ is a $k$-chain, $\partial_k(C)$ is the $(k-1)$-chain consisting of $(k-1)$-faces that occur an odd number of times on the boundary of $k$-faces in $C$. It follows directly from its action on convex polytopes that, for any $C$, $\partial_{k-1}\bigl(\partial_k(C)\bigr)=\emptyset$: the boundary of a boundary is empty. Conversely, if $B$ is a $(k-1)$-chain with empty boundary -- that is, if $\partial_{k-1}(B)=\emptyset$ -- then $B$ is itself a boundary: there exists $C$ such that $B=\partial_k(C)$. In other spaces than $\mathbb{R}^d$, there can exist chains with empty boundary that are not themselves boundaries, corresponding to nontrivial elements of the homology groups of the space, but $\mathbb{R}^d$ has trivial homology so such chains do not exist.

For a set of $d$-dimensional training points with classifications, we consider a subgraph of the Delaunay graph consisting of the edges connecting cells with the same classification as each other, and we define a \emph{decision component} to be a connected component of this subgraph; see \cref{fig:decision-boundaries}. As a set of $d$-dimensional Voronoi cells, it can be considered as a $k$-chain.
If $C$ is a decision component, form a graph $B_C$ whose vertices are the Voronoi walls of $\partial_k(C)$, with two vertices adjacent when the two walls they represent meet in a junction. We define a \emph{decision boundary component} to be a connected component of $B_C$ for any decision component $C$. Because it represents a set of Voronoi walls, it can be considered as a $(k-1)$-chain.

\begin{lemma}
\label{lem:boundaries-separate}
Every decision boundary component is the boundary of a set of Voronoi cells.
\end{lemma}

\begin{proof}
Because $\partial_{k-1}\bigl(\partial_k(C)\bigr)=\emptyset$, and a set of walls has an empty boundary if and only if they touch each junction an even number of times, it follows that all vertices in $B_C$ have even degree, and therefore that the same thing is true in every decision boundary component. Therefore, if $D$ is a decision boundary component, $\partial_{k-1}(D)=\emptyset$. It follows from the triviality of the homology of $\mathbb{R}^d$ that there exists a $k$-chain $E$ such that $\partial_k(E)=D$.
\end{proof}

If $D$ is a decision boundary component, we define a \emph{side} of $D$ to be a set of Voronoi cells having $D$ as its boundary. By \cref{lem:boundaries-separate}, every decision boundary component has at least one side. (Actually, at least two, because the complement of a side is another side.)

\section{Correctness}

\subsection{Only relevant points are found}

\cref{alg:relevant} adds points to its set $R$ of relevant points in two ways: by finding endpoints of minimum spanning tree edges and by finding extreme points of inverted point sets. We prove in this section that in both cases the points that are found are truly relevant: there is a nearest-neighbor classification query that would be answered incorrectly if any one of these points were omitted.

\begin{lemma}
\label{lem:mst-endpoints-are-relevant}
If $p$ is an endpoint of an edge of a minimum spanning tree for which the other endpoint $p'$ has a different classification, then $p$ is relevant with respect to the whole data set, and also relevant with respect to any subset of the data set that includes both $p$ and $p'$.
\end{lemma}

\begin{proof}
By \cref{lem:mst-is-delaunay}, there is a Voronoi wall between the cells for $p$ and $p'$. Let $q$ be any point interior to that wall, and perturb $q$ to a point $q'$ closer to $p$ than to $p'$, by a perturbation sufficiently small that it does not touch or cross any other Voronoi walls. Then the correct nearest neighbor classification of $q'$ is the same as that for $p$, but if $p$ were omitted then the classification would become the same as for $p'$, a different value. Therefore $p$ is relevant.
\end{proof}

\begin{lemma}
\label{lem:extreme-points-are-relevant}
If $r$ is relevant for the whole data set, $p$ has a different classification than $r$, and $p$ corresponds to one of the extreme points of the set $S_r$ constructed by \cref{alg:relevant}, then $p$ is relevant with respect to the whole data set, and also relevant with respect to any subset of the data set that includes both $r$ and $p$.
\end{lemma}

\begin{proof}
By \cref{lem:extreme-wall}, $r$ and $p$ share a wall in the Voronoi diagram of $S_r$. Let $w$ be a point in the relative interior of this wall, chosen in general position so that, in the Voronoi diagram of the whole training set, line segment $pw$ does not cross any faces of lower dimension than a wall. With this choice, $w$ is equidistant from $r$ and $p$, and farther from all other points of different classification to $r$.

Within line segment $pw$, each point of the line segment has $p$ as the nearest neighbor among points with a different classification to $r$; however, some points of the line segment may have an even closer neighbor that has the same classification as~$r$. (For instance, this happens for the leftmost red Voronoi neighbor of $r$ in \cref{fig:steps}.) Let $x$ be a point of line segment $pw$ that is equidistant from $p$ and from the nearest training point to $x$ with the same classification as $r$. The existence of $x$ can be seen from the intermediate value theorem, noting that at one endpoint of this segment, $p$, $p$ is closer than any other training point, while at the other endpoint, $w$, the nearest training point with the same classification as $r$ is at equal or closer distance than $p$. Because of our choice of $w$ as being in general position, no other training point can be as close to $x$ as these two points. Then $x$ lies on a Voronoi wall between $p$ and a point with the same classification as $r$, so $p$ is relevant by the same perturbation argument as in the proof of \cref{lem:mst-endpoints-are-relevant}.
\end{proof}

\begin{lemma}
\label{lem:all-are-relevant}
All points identified as relevant by \cref{alg:relevant} are relevant, both with respect to the whole data set and with respect to the subset of points identified by the algorithm.
\end{lemma}

\begin{proof}
This follows by induction on the number of iterations of the outer loop of \cref{alg:relevant}, using \cref{lem:mst-endpoints-are-relevant} as the base case and \cref{lem:extreme-points-are-relevant} for the induction step.
\end{proof}

\subsection{All relevant points are found}

If $D$ is a decision boundary component, and $w$ is a wall in $D$, between the Voronoi cells for points $p$ and $q$, we say that $p$ and $q$ are \emph{defining points} of $D$.

\begin{lemma}
\label{lem:at-least-one}
For every decision boundary component $D$, \cref{alg:relevant} identifies at least one defining point of $D$ as relevant.
\end{lemma}

\begin{proof}
By \cref{lem:boundaries-separate}, $D$ has a side $E$. Because the minimum spanning tree connects all the Voronoi cells, it includes at least one edge $e$ that connects a cell in $E$ to another cell not in $E$. The wall between these two cells is part of the boundary of $E$, so it belongs to $D$, and its endpoints are defining points of $D$. Because this wall belongs to a decision boundary component, it separates two cells of different classifications, so the endpoints of $e$ will be identified as relevant by the phase of \cref{alg:relevant} that finds endpoints of minimum spanning tree edges whose endpoints have different classifications.
\end{proof}

Once we have identified at least one defining point of a decision boundary component, the second phase of \cref{alg:relevant} finds all of them, as the following lemmas show.

\begin{lemma}
\label{lem:both-defining}
If \cref{alg:relevant} identifies one of the two defining points of a wall of a decision boundary component as relevant, it identifies the other defining point of the same wall.
\end{lemma}

\begin{proof}
Let $p$ be the first of the two defining points to be identified, and $q$ be the other defining point for the same wall. Because they are separated by a decision boundary component, they have different classifications. Then $q$ is a neighbor of $p$ in the Delaunay graph of all of the training points, and therefore also in the Delaunay graph of $p$ and the subset of training points with different classifications to $p$. Therefore, by \cref{lem:extreme-wall}, $q$ will be found as the inverted image of one of the extreme points in $S_p$.
\end{proof}

\begin{lemma}
\label{lem:consecutive-walls}
Let $w$ and $w'$ be walls sharing a junction, such that both $w$ and $w'$ separate cells with different classifications and such that all cells between them (within one of the two angles that they form at their shared junction) have a single classification. Suppose also that \cref{alg:relevant} identifies one of the defining points of $w$ as relevant. Then it also identifies one of the defining points of $w'$ as relevant.
\end{lemma}

\begin{proof}
By \cref{lem:both-defining}, \cref{alg:relevant} identifies the defining point $p$ of $w$ that lies within the angle between $w$ and $w'$. Let $q$ be the defining point of $w'$ that lies outside this angle, let $J$ be the junction of $w$ and $w'$, and consider the link of this junction, within a plane perpendicular to $J$. Within this link, the Voronoi cells incident to $J$ divide the plane into wedges, meeting at the point where $J$ crosses the plane of the link. The defining sites of these cells are all closer to this point than any other training points. For any subset of the training points that includes at least one of the defining sites of a cell incident to $J$, the cross-section of the Voronoi diagram in the plane of the link will still have the structure of a set of wedges for the remaining sites, in the same circular ordering. In particular, in the Voronoi diagram of $S_p$, the wedges of $p$ and $q$ will be consecutive. Therefore, these two cells are separated by a wall in the Voronoi diagram of $S_p$, and by \cref{lem:extreme-wall}, $q$ will be found as the inverted image of one of the extreme points in $S_p$.
\end{proof}

\begin{lemma}
\label{lem:relevant-are-found}
\cref{alg:relevant} identifies all relevant points with respect to the whole training set, and with respect to the set of points that it identies.
\end{lemma}

\begin{proof}
Let $r$ be a training point that is relevant with respect to the whole training set, and let $q$ be a query point that would get the wrong classification if $r$ is removed from the training set (witnessing the relevance of $r$). Then $r$ must be the nearest training point to $q$, and the second-nearest training point $p$ to $q$ must have a different classification than $r$. Because $p$ and $r$ are the nearest and second-nearest points to $q$, they must share a Voronoi wall~$w$, which belongs to a decision boundary component $D$ of the decision component of $r$. By \cref{lem:at-least-one}, \cref{alg:relevant} identifies at least one defining point of $D$, and by \cref{lem:consecutive-walls} the defining points that it identifies extend from any wall to any other wall of $D$ that is consecutive at a junction, and therefore also to any other wall of $D$ that is adjacent at a junction, and to any other wall that is connected through a sequence of adjacencies at junctions. But $D$ was defined as a set of walls that are adjacent in this way, so \cref{alg:relevant} identifies at least one defining point of $w$. By \cref{lem:both-defining} it identifies both defining points, and therefore it identifies~$w$.

Removing an irrelevant point does not change the decision boundary components or the defining points of their walls, so it does not create new relevant points. Therefore, the relevant points with respect to the whole training set identified by \cref{alg:relevant} are also all of the relevant points with respect to the set of points that it identifies.
\end{proof}

\begin{theorem}
\label{thm:correctness}
The nearest-neighbor classifications obtained from the set of relevant points identified by \cref{alg:relevant} equal the nearest-neighbor classifications obtained from the whole training set.
\end{theorem}

\begin{proof}
The classification of any point $p$ of $\mathbb{R}^d$, for a given set of labeled points, may be obtained from the label of its nearest relevant point $r$ for the given set; all Voronoi walls that separate $p$ from $r$ must have the same label on both sides, for otherwise the two points defining the wall would be relevant and at least one would be nearer to $p$ than to $r$, from which it follows that the nearest neighbor to $p$ has the same label as $r$. By \cref{lem:relevant-are-found} the identity of the relevant points does not change between the whole training points and the subset of points identified by the algorithm, and therefore the identity of the nearest relevant point cannot change.
\end{proof}

\section{Analysis}

We are now ready to prove our main results on the time bounds for finding relevant points.

\begin{theorem}
\label{thm:simple-simple}
\cref{alg:relevant}, implemented using Jarník's or Borůvka's algorithm for its minimum spanning trees, \cref{alg:simple-extreme} for extreme points, and Seidel's algorithm~\cite{Sei-DCG-91} for the linear programming steps of \cref{alg:simple-extreme}, finds all relevant points of a given training set of size $n$, having $k$ relevant points, in any constant dimension $d$, in total time $O(n^2+k^2n)$.
\end{theorem}

\begin{proof}
The correctness of the algorithm is shown by \cref{lem:all-are-relevant} and \cref{lem:relevant-are-found}. By \cref{lem:all-are-relevant}, there are $k$ extreme-point computations, each of which identifies at most $k$ extreme points. The time bounds for the overall algorithm are obtained simply by plugging in the time bounds for its subroutines.
\end{proof}

\begin{theorem}
\cref{alg:relevant}, implemented using the algorithm of Agarwal, Edelsbrunner, Schwarzkopf, and Welzl~\cite{AgaEdeSch-DCG-91} for its minimum spanning trees, and the algorithm of Chan~\cite{Cha-DCG-96b} for its extreme points, finds all relevant points of a given training set of size $n$, having $k$ relevant points, in randomized expected time $O\bigl((n\log n)^{4/3}+kn\log k\bigr)$ for $d=3$ and in time
\[
O\left(n^{2-\frac{2}{\lceil d/2\rceil+1}+\varepsilon}
+k(nk)^{1-\frac{1}{\lfloor d/2\rfloor+1}}(\log n)^{O(1)}\right).
\]
for all constant dimensions $d>3$.
\end{theorem}

The proof is the same as for \cref{thm:simple-simple}.

\begin{theorem}
\cref{alg:relevant}, implemented using Jarník's or Borůvka's algorithm for its minimum spanning trees, \cref{alg:simple-extreme} for extreme points, and an arbitrary linear programming algorithm for the linear programming steps of \cref{alg:simple-extreme}, finds all relevant points of a given training set of size $n$, having $k$ relevant points, in any constant dimension $d$, in time $O\bigl(n^2+k^2n+knL(d,k)\bigr)$, where $L(x,y)$ denotes the time to solve a linear program of dimension (number of variables) $x$ and size (number of constraints) $y$.
\end{theorem}

Again, the proof is the same.

\bibliographystyle{plainurl}
\bibliography{voronoi}

\begin{thebibliography}{10}

\bibitem{AgaEdeSch-DCG-91}
Pankaj~K. Agarwal, Herbert Edelsbrunner, Otfried Schwarzkopf, and Emo Welzl.
\newblock {Euclidean minimum spanning trees and bichromatic closest pairs}.
\newblock {\em Discrete Comput. Geom.}, 6(5):407{--}422, 1991.
\newblock \href {http://dx.doi.org/10.1007/BF02574698}
  {\path{doi:10.1007/BF02574698}}.

\bibitem{AgaMat-Algo-95}
Pankaj~K. Agarwal and Ji{\v{r}}{\'\i} Matou{\v{s}}ek.
\newblock {Dynamic half-space range reporting and its applications}.
\newblock {\em Algorithmica}, 13(4):325{--}345, 1995.
\newblock \href {http://dx.doi.org/10.1007/BF01293483}
  {\path{doi:10.1007/BF01293483}}.

\bibitem{AndInd-CACM-08}
Alexandr Andoni and Piotr Indyk.
\newblock {Near-optimal hashing algorithms for approximate nearest neighbor in
  high dimensions}.
\newblock {\em Commun. ACM}, 51(1):117{--}122, 2008.
\newblock \href {http://dx.doi.org/10.1145/1327452.1327494}
  {\path{doi:10.1145/1327452.1327494}}.

\bibitem{AryMalMou-JACM-10}
Sunil Arya, Theocharis Malamatos, and David~M. Mount.
\newblock {Space-time tradeoffs for approximate nearest neighbor searching}.
\newblock {\em J. ACM}, 57(1):A1:1{--}A1:54, 2010.
\newblock \href {http://dx.doi.org/10.1145/1613676.1613677}
  {\path{doi:10.1145/1613676.1613677}}.

\bibitem{AryMouNet-JACM-98}
Sunil Arya, David~M. Mount, Nathan~S. Netanyahu, Ruth Silverman, and Angela~Y.
  Wu.
\newblock {An optimal algorithm for approximate nearest neighbor searching in
  fixed dimensions}.
\newblock {\em J. ACM}, 45(6):891{--}923, 1998.
\newblock \href {http://dx.doi.org/10.1145/293347.293348}
  {\path{doi:10.1145/293347.293348}}.

\bibitem{AurKleLee-13}
Franz Aurenhammer, Rolf Klein, and Der-Tsai Lee.
\newblock {\em {Voronoi Diagrams and Delaunay Triangulations}}.
\newblock World Scientific, 2013.
\newblock \href {http://dx.doi.org/10.1142/8685} {\path{doi:10.1142/8685}}.

\bibitem{BreDemEri-DCG-05}
David Bremner, Erik Demaine, Jeff Erickson, John Iacono, Stefan Langerman, Pat
  Morin, and Godfried Toussaint.
\newblock {Output-sensitive algorithms for computing nearest-neighbour decision
  boundaries}.
\newblock {\em Discrete Comput. Geom.}, 33(4):593{--}604, 2005.
\newblock \href {http://dx.doi.org/10.1007/s00454-004-1152-0}
  {\path{doi:10.1007/s00454-004-1152-0}}.

\bibitem{CalKos-SODA-93}
Paul~B. Callahan and S.~Rao Kosaraju.
\newblock {Faster algorithms for some geometric graph problems in higher
  dimensions}.
\newblock In {\em Proc. 4th Symp. Discrete Algorithms (SODA 1993)}, pages
  291{--}300. ACM, 1993.

\bibitem{Cha-DCG-96a}
Timothy~M. Chan.
\newblock {Optimal output-sensitive convex hull algorithms in two and three
  dimensions}.
\newblock {\em Discrete Comput. Geom.}, 16(4):361{--}368, 1996.
\newblock \href {http://dx.doi.org/10.1007/BF02712873}
  {\path{doi:10.1007/BF02712873}}.

\bibitem{Cha-DCG-96b}
Timothy~M. Chan.
\newblock {Output-sensitive results on convex hulls, extreme points, and
  related problems}.
\newblock {\em Discrete Comput. Geom.}, 16(4):369{--}387, 1996.
\newblock \href {http://dx.doi.org/10.1007/BF02712874}
  {\path{doi:10.1007/BF02712874}}.

\bibitem{Cha-DCG-98}
Timothy~M. Chan.
\newblock {Approximate nearest neighbor queries revisited}.
\newblock {\em Discrete Comput. Geom.}, 20(3):359{--}373, 1998.
\newblock \href {http://dx.doi.org/10.1007/PL00009390}
  {\path{doi:10.1007/PL00009390}}.

\bibitem{Cha-TALG-18}
Timothy~M. Chan.
\newblock {Improved deterministic algorithms for linear programming in low
  dimensions}.
\newblock {\em ACM Trans. Algorithms}, 14(3):A30:1{--}A30:10, 2018.
\newblock \href {http://dx.doi.org/10.1145/3155312}
  {\path{doi:10.1145/3155312}}.

\bibitem{Cha-DCG-93}
Bernard Chazelle.
\newblock {An optimal convex hull algorithm in any fixed dimension}.
\newblock {\em Discrete Comput. Geom.}, 10(4):377{--}409, 1993.
\newblock \href {http://dx.doi.org/10.1007/BF02573985}
  {\path{doi:10.1007/BF02573985}}.

\bibitem{ChaMat-CG-95}
Bernard Chazelle and Ji{\v{r}}{\'\i} Matou{\v{s}}ek.
\newblock {Derandomizing an output-sensitive convex hull algorithm in three
  dimensions}.
\newblock {\em Comput. Geom.}, 5(1):27{--}32, 1995.
\newblock \href {http://dx.doi.org/10.1016/0925-7721(94)00018-Q}
  {\path{doi:10.1016/0925-7721(94)00018-Q}}.

\bibitem{Cla-FOCS-94}
Kenneth~L. Clarkson.
\newblock {More output-sensitive geometric algorithms}.
\newblock In {\em Proc. 35th Symp. Foundations of Computer Science (FOCS
  1994)}, pages 695{--}702. IEEE Computer Society, 1994.
\newblock \href {http://dx.doi.org/10.1109/SFCS.1994.365723}
  {\path{doi:10.1109/SFCS.1994.365723}}.

\bibitem{Cla-JACM-95}
Kenneth~L. Clarkson.
\newblock {Las Vegas algorithms for linear and integer programming when the
  dimension is small}.
\newblock {\em J. ACM}, 42(2):488{--}499, 1995.
\newblock \href {http://dx.doi.org/10.1145/201019.201036}
  {\path{doi:10.1145/201019.201036}}.

\bibitem{ClaSho-DCG-89}
Kenneth~L. Clarkson and Peter~W. Shor.
\newblock {Applications of random sampling in computational geometry. II}.
\newblock {\em Discrete Comput. Geom.}, 4(5):387{--}421, 1989.
\newblock \href {http://dx.doi.org/10.1007/BF02187740}
  {\path{doi:10.1007/BF02187740}}.

\bibitem{CovHar-TIT-67}
T.~Cover and P.~Hart.
\newblock {Nearest neighbor pattern classification}.
\newblock {\em IEEE Transactions on Information Theory}, 13(1):21{--}27, 1967.
\newblock \href {http://dx.doi.org/10.1109/tit.1967.1053964}
  {\path{doi:10.1109/tit.1967.1053964}}.

\bibitem{DatImmInd-SoCG-04}
Mayur Datar, Nicole Immorlica, Piotr Indyk, and Vahab~S. Mirrokni.
\newblock {Locality-sensitive hashing scheme based on $p$-stable
  distributions}.
\newblock In Jack Snoeyink and Jean-Daniel Boissonnat, editors, {\em Proc. 20th
  Symp. Computational Geometry (SoCG 2004)}, pages 253{--}262. ACM, 2004.
\newblock \href {http://dx.doi.org/10.1145/997817.997857}
  {\path{doi:10.1145/997817.997857}}.

\bibitem{DewVra-UM-77}
A.~K. Dewdney and J.~K. Vranch.
\newblock {A convex partition of $\mathbb{R}^3$ with applications to {C}rum's
  problem and {K}nuth's post-office problem}.
\newblock {\em Utilitas Math.}, 12:193{--}199, 1977.

\bibitem{Dwy-DCG-91}
Rex~A. Dwyer.
\newblock {Higher-dimensional Voronoi diagrams in linear expected time}.
\newblock {\em Discrete Comput. Geom.}, 6(4):343{--}367, 1991.
\newblock \href {http://dx.doi.org/10.1007/BF02574694}
  {\path{doi:10.1007/BF02574694}}.

\bibitem{DyeFri-MP-89}
Martin~E. Dyer and Alan~M. Frieze.
\newblock {A randomized algorithm for fixed-dimensional linear programming}.
\newblock {\em Math. Programming}, 44(2, (Ser. A)):203{--}212, 1989.
\newblock \href {http://dx.doi.org/10.1007/BF01587088}
  {\path{doi:10.1007/BF01587088}}.

\bibitem{Epp-HCG-00}
David Eppstein.
\newblock {Spanning trees and spanners}.
\newblock In J{\"o}rg-Rudiger Sack and Jorge Urrutia, editors, {\em Handbook of
  Computational Geometry}, pages 425{--}461. Elsevier, 2000.
\newblock \href {http://dx.doi.org/10.1016/B978-044482537-7/50010-3}
  {\path{doi:10.1016/B978-044482537-7/50010-3}}.

\bibitem{EppGooSun-IJCGA-08}
David Eppstein, Michael~T. Goodrich, and Jonathan~Z. Sun.
\newblock {Skip quadtrees: dynamic data structures for multidimensional point
  sets}.
\newblock {\em Internat. J. Comput. Geom. Appl.}, 18(1-2):131{--}160, 2008.
\newblock \href {http://dx.doi.org/10.1142/S0218195908002568}
  {\path{doi:10.1142/S0218195908002568}}.

\bibitem{GioIndMot-VLDB-99}
Aristides Gionis, Piotr Indyk, and Rajeev Motwani.
\newblock {Similarity search in high dimensions via hashing}.
\newblock In Malcolm~P. Atkinson, Maria~E. Or{\l}owska, Patrick Valduriez,
  Stanley~B. Zdonik, and Michael~L. Brodie, editors, {\em Proc. 25th Int. Conf.
  Very Large Data Bases (VLDB 1999)}, pages 518{--}529. Morgan Kaufmann, 1999.
\newblock URL: \url{https://www.vldb.org/conf/1999/P49.pdf}.

\bibitem{Gru-CP-03}
Branko Gr{\"u}nbaum.
\newblock {\em {Convex Polytopes}}, volume 221 of {\em Graduate Texts in
  Mathematics}.
\newblock Springer, 2nd edition, 2003.
\newblock See in particular Gr{\"u}nbaum{'}s discussion of the Perles
  configuration on pp. 93{--}94.

\bibitem{GuiKnuSha-Algo-92}
Leonidas~J. Guibas, Donald~E. Knuth, and Micha Sharir.
\newblock {Randomized incremental construction of Delaunay and Voronoi
  diagrams}.
\newblock {\em Algorithmica}, 7(4):381{--}413, 1992.
\newblock \href {http://dx.doi.org/10.1007/BF01758770}
  {\path{doi:10.1007/BF01758770}}.

\bibitem{Hen-94}
Michael Henle.
\newblock {\em {A Combinatorial Introduction to Topology}}.
\newblock Dover Publications, 1994.

\bibitem{IndMot-STOC-98}
Piotr Indyk and Rajeev Motwani.
\newblock {Approximate nearest neighbors: towards removing the curse of
  dimensionality}.
\newblock In Jeffrey~Scott Vitter, editor, {\em Proc. 30th Symp. Theory of
  Computing (STOC 1998)}, pages 604{--}613. ACM, 1998.
\newblock \href {http://dx.doi.org/10.1145/276698.276876}
  {\path{doi:10.1145/276698.276876}}.

\bibitem{KarMan-Comb-20}
C.~S. Karthik and Pasin Manurangsi.
\newblock {On closest pair in Euclidean metric: monochromatic is as hard as
  bichromatic}.
\newblock {\em Combinatorica}, 40(4):539{--}573, 2020.
\newblock \href {http://dx.doi.org/10.1007/s00493-019-4113-1}
  {\path{doi:10.1007/s00493-019-4113-1}}.

\bibitem{KirSei-SICOMP-86}
David~G. Kirkpatrick and Raimund Seidel.
\newblock {The ultimate planar convex hull algorithm?}
\newblock {\em SIAM J. Comput.}, 15(1):287{--}299, 1986.
\newblock \href {http://dx.doi.org/10.1137/0215021}
  {\path{doi:10.1137/0215021}}.

\bibitem{Kle-AM-80}
Victor Klee.
\newblock {On the complexity of $d$-dimensional {V}oronoi diagrams}.
\newblock {\em Arch. Math.}, 34(1):75{--}80, 1980.
\newblock \href {http://dx.doi.org/10.1007/BF01224932}
  {\path{doi:10.1007/BF01224932}}.

\bibitem{KrzLevNil-NJC-99}
Drago Krznaric, Christos Levcopoulos, and Bengt~J. Nilsson.
\newblock {Minimum spanning trees in $d$ dimensions}.
\newblock {\em Nordic J. Comput.}, 6(4):446{--}461, 1999.

\bibitem{Mat-JCSS-95}
Ji{\v{r}}{\'\i} Matou{\v{s}}ek.
\newblock {Approximations and optimal geometric divide-and-conquer}.
\newblock {\em J. Comput. System Sci.}, 50(2):203{--}208, 1995.
\newblock \href {http://dx.doi.org/10.1006/jcss.1995.1018}
  {\path{doi:10.1006/jcss.1995.1018}}.

\bibitem{MatShaWel-Algo-96}
Ji{\v{r}}{\'\i} Matou{\v{s}}ek, Micha Sharir, and Emo Welzl.
\newblock {A subexponential bound for linear programming}.
\newblock {\em Algorithmica}, 16(4-5):498{--}516, 1996.
\newblock \href {http://dx.doi.org/10.1007/BF01940877}
  {\path{doi:10.1007/BF01940877}}.

\bibitem{Meg-JACM-84}
Nimrod Megiddo.
\newblock {Linear programming in linear time when the dimension is fixed}.
\newblock {\em J. ACM}, 31(1):114{--}127, 1984.
\newblock \href {http://dx.doi.org/10.1145/2422.322418}
  {\path{doi:10.1145/2422.322418}}.

\bibitem{OttSchSou-NJC-01}
Thomas Ottmann, Sven Schuierer, and Subbiah Soundaralakshmi.
\newblock {Enumerating extreme points in higher dimensions}.
\newblock {\em Nordic J. Comput.}, 8(2):179{--}192, 2001.

\bibitem{PreTam-SICOMP-92}
Franco~P. Preparata and Roberto Tamassia.
\newblock {Efficient point location in a convex spatial cell-complex}.
\newblock {\em SIAM J. Comput.}, 21(2):267{--}280, 1992.
\newblock \href {http://dx.doi.org/10.1137/0221020}
  {\path{doi:10.1137/0221020}}.

\bibitem{Sei-AGDM-91}
Raimund Seidel.
\newblock {Exact upper bounds for the number of faces in $d$-dimensional
  {V}oronoi diagrams}.
\newblock In {\em Applied Geometry and Discrete Mathematics}, volume~4 of {\em
  DIMACS Ser. Discrete Math. Theoret. Comput. Sci.}, pages 517{--}529. Amer.
  Math. Soc., 1991.

\bibitem{Sei-DCG-91}
Raimund Seidel.
\newblock {Small-dimensional linear programming and convex hulls made easy}.
\newblock {\em Discrete Comput. Geom.}, 6(5):423{--}434, 1991.
\newblock \href {http://dx.doi.org/10.1007/BF02574699}
  {\path{doi:10.1007/BF02574699}}.

\bibitem{ShaHoe-FOCS-75}
Michael~Ian Shamos and Dan Hoey.
\newblock {Closest-point problems}.
\newblock In {\em Proc. 16th Symp. Foundations of Computer Science (FOCS
  1975)}, pages 151{--}162. IEEE Computer Society, 1975.
\newblock \href {http://dx.doi.org/10.1109/SFCS.1975.8}
  {\path{doi:10.1109/SFCS.1975.8}}.

\bibitem{Wat-CJ-81}
D.~F. Watson.
\newblock {Computing the $n$-dimensional Delaunay tessellation with application
  to Voronoi polytopes}.
\newblock {\em Comput. J.}, 24(2):167{--}172, 1981.
\newblock \href {http://dx.doi.org/10.1093/comjnl/24.2.167}
  {\path{doi:10.1093/comjnl/24.2.167}}.

\end{thebibliography}
\end{document}